\documentclass[preprint,11pt]{imsart}
\pdfoutput=1
\RequirePackage[OT1]{fontenc}
\usepackage{amsthm,amsmath,amssymb,natbib,bm,enumerate}
\RequirePackage[colorlinks,citecolor=blue,urlcolor=blue]{hyperref}
\startlocaldefs
\numberwithin{equation}{section}
\theoremstyle{plain}
\newtheorem{thm}{Theorem}[section]
\newtheorem{lemma}{Lemma}[section]
\newtheorem{algorithm}{Algorithm}[section]
\theoremstyle{remark}
\newtheorem{remark}{Remark}[section]
\endlocaldefs

\def\citeapos#1{\citeauthor{#1}'s (\citeyear{#1})}

\allowdisplaybreaks

\begin{document}

\begin{frontmatter}
\title{A new Monte Carlo sampling in Bayesian probit regression}
\runtitle{A new Monte Carlo sampling}
%\thankstext{T1}{Footnote to the title with the `thankstext' command.}

\begin{aug}
\author{\fnms{Yuzo} \snm{Maruyama}\thanksref{t1,m1}
\ead[label=e1]
{maruyama@csis.u-tokyo.ac.jp}}
\and
\author{\fnms{William, E.} \snm{Strawderman}
\thanksref{t2,m2}
\ead[label=e2]{straw@stat.rutgers.edu}}

\thankstext{t1}{This work was partially supported by KAKENHI \#23740067.}
\thankstext{t2}{This work was partially supported by a grant from the Simons Foundation (\#209035 to William Strawderman).}
\address{University of Tokyo\thanksmark{m1} and Rutgers University\thanksmark{m2} \\
\printead{e1,e2}}
\runauthor{Y. Maruyama and W. Strawderman}

%\affiliation{Some University and Another University}
\end{aug}

\begin{abstract}
We study probit regression from a Bayesian perspective and give an alternative
form for the posterior distribution when the prior distribution for the regression
parameters is the uniform distribution. This new form allows simple
Monte Carlo simulation of the posterior as opposed to MCMC simulation studied
in much of the literature and may therefore 
be more efficient computationally. 
We also provide alternative explicit expression for the first and second moments.
Additionally we provide analogous results for Gaussian priors.
%We also provide a similar development for normal priors. A simulation study
%indicates that this new method performs well in a simulation study and in an
%example study.
\end{abstract}


\end{frontmatter}

\section{Introduction}
\label{sec:intro}
The analysis of binary response data is important 
in statistics and related areas including econometrics and biometrics.
The classical maximum likelihood method 
and inferences based on the associated asymptotic theory
is often not accurate for small sample sizes.
A Bayesian approach with respect to the non-informative
or flat prior is a worthy natural competitor.

We study this problem with the aim of providing an alternative
expression for the posterior distribution that allows simple Monte Carlo
simulation as opposed to the somewhat more involved MCMC methods
in much of the literature. (See e.g.~\cite{Albert-Chib-1993})
We also give explicit expressions for the first and second moments
of the regression parameters.
Additionally we provide analogous results for Gaussian priors.

Suppose that $n$ independent binary random variables 
$Y_1,\dots,Y_n$ are observed, where $Y_i=1$ 
with probability of success $p_i$. 
The $p_i$ are related to a set of
covariates that may be continuous or discrete. Define the probit
regression model as $p_i=\Phi(\bm{x}'_i\bm{\beta})$, $i=1,\dots,n$,
where $\bm{\beta}$ is a $p\times 1$ vector of unknown parameters,
$\bm{x}_i$ is a vector of known covariates, and $\Phi$ is the
standard Gaussian cumulative distribution function. 
Let 
\begin{equation*}\label{data_y}
\bm{Y}=(Y_1,\dots,Y_n)', \ \bm{y}=(y_1,\dots, y_n)'
\end{equation*}
and 
\begin{equation*}\label{data_x}
 \bm{X}=(\bm{x}_1,\dots,\bm{x}_n)'
\end{equation*} 
with full rank $p$.
Then the joint probability distribution of $\bm{y}$ is given by
\begin{equation}\label{joint_1}
\mathrm{Pr}(\bm{Y}=\bm{y}|\bm{\beta})= \prod_{i=1}^n \Phi(\bm{x}'_i\bm{\beta})^{y_i}\left[1-\Phi(\bm{x}'_i\bm{\beta})\right]^{1-y_i}.
\end{equation}

The posterior density with respect to the flat prior $\pi(\bm{\beta})=1$ is
 proportional to the joint density \eqref{joint_1}.
Because this posterior is somewhat intractable theoretically,
instead of pursuing analytic results, a variety of simulation algorithms
have been proposed for obtaining samples from the posterior distribution $\pi(\bm{\beta}|\bm{y})$.
To our knowledge, the default choice is the so called \citeapos{Albert-Chib-1993} sampler, which we now describe.

The computational scheme proceeds by introducing $n$ independent latent
variables $Z_1,\dots,Z_n$, where $Z_i\sim N(\bm{x}'_i\bm{\beta},1)$.
If we let $Y_i = I(Z_i > 0)$, then $Y_1,\dots,Y_n$ 
are independent Bernoulli with $p_i = P(Y_i = 1) = \Phi(\bm{x}'_i\bm{\beta})$.
Under the flat prior, the posterior density of $\bm{\beta}$ and 
$\bm{Z}=(Z_1,\dots,Z_n)$ given $\bm{y}=(y_1,\dots,y_n)$ is
\begin{equation*}\label{bz|y}
\begin{split}
& \pi(\bm{\beta},\bm{Z}|\bm{y}) \\
&=
\textstyle{\prod_{i=1}^n}\left\{I(Z_i>0)I(y_i=1)+I(Z_i\leq 0)I(y_i=0)\right\}
\phi(Z_i-\bm{x}'_i\bm{\beta})
\end{split} 
\end{equation*}
where $\phi$ is the standard Gaussian probability density function.
Since $  \pi(\bm{\beta}|\{\bm{y},\bm{z}\})$ is proportional to
\begin{equation*}\label{b|yz}
% \pi(\bm{\beta}|\{\bm{y},\bm{Z}\})=
\textstyle{\prod_{i=1}^n}
\phi(z_i-\bm{x}'_i\bm{\beta}),
\end{equation*}
we have
\begin{equation}\label{b|yz1}
 \bm{\beta}|\{\bm{y},\bm{z}\}\sim N_p((\bm{X}'\bm{X})^{-1}\bm{X}'\bm{z},(\bm{X}'\bm{X})^{-1}).
\end{equation}
Further it is clear that
\begin{equation}\label{z|by}
 z_i|\{y_i,\bm{\beta}\}\sim
\begin{cases}
 N_{+}(\bm{x}'_i\bm{\beta},1,0) & \mbox{ if }y_i=1,\\
 N_{-}(\bm{x}'_i\bm{\beta},1,0) & \mbox{ if }y_i=0,
\end{cases}
\end{equation}
where $N_{+}(\nu,1,0)$ and $N_{+}(\nu,1,0)$ are the Gaussian distributions with
mean $\nu$ and variance $1$ that are left-truncated and right-truncated at $0$,
respectively.
Based on \eqref{b|yz1} and \eqref{z|by}, 
the corresponding Gibbs sampler is derived
and $\pi(\bm{\beta}|\bm{y})$ is approximated.
For discussions of related computational techniques for probit regression,
see \cite{Marin-Robert-2007}.

In this paper, we push the theoretical analysis somewhat further
to provide an expression for the posterior 
that allows more direct simulation and thereby, 
a more efficient random sampler.
Let 
\begin{equation}\label{X_y}
 \bm{X}_{\bm{y}}=\left\{2\mathrm{diag}(\bm{y})-\bm{I}_n\right\}\bm{X},
\end{equation}
where $\mathrm{diag}(\bm{y})$ is the $n\times n$ matrix with $y_i$
as the $i$-th diagonal entry, and
which will be seen to be sufficient for the joint probability.
Also let $ \bm{\Psi}[\bm{X}_{\bm{y}}]$ be the projection matrix onto the orthogonal
complement of the column space of $\bm{X}_{\bm{y}}$, which is
\begin{equation}\label{Psi_intro}
 \bm{\Psi}[\bm{X}_{\bm{y}}]
=\bm{I}-\bm{X}_{\bm{y}}(\bm{X}'_{\bm{y}}\bm{X}_{\bm{y}})^{-1}\bm{X}'_{\bm{y}}.
\end{equation}
Then under the mild condition on $ \bm{\Psi}[\bm{X}_{\bm{y}}]$, 
 $\pi(\bm{\beta}|\bm{y})$ is shown to be
\begin{equation}
\pi(\bm{\beta}|\bm{y})=
\iint \pi(\bm{\beta}|s,\bm{u},\bm{y})\pi(s|\bm{u},\bm{y})\pi(\bm{u}|\bm{y})
dsd\bm{u}
\end{equation}
where the elements of this hierarchical structure are given by
\begin{equation}\label{mc_step_0_intro}
\begin{split}
 \pi(\bm{\beta}|s,\bm{u},\bm{y})&=N_p((\bm{X}'\bm{X})^{-1}\bm{X}'_{\bm{y}}s^{1/2}\bm{u},
(\bm{X}'\bm{X})^{-1}),\\
\pi(s|\bm{u},\bm{y}) &=\|\bm{\Psi}[\bm{X}_{\bm{y}}]\bm{u}\|^{-2}\chi_n^2, \\
\pi(\bm{u}|\bm{y})& \propto \|\bm{\Psi}[\bm{X}_{\bm{y}}]\bm{u}\|^{-n} \text{ on }\mathcal{S}_+^n,
%\pi(\bm{h})& \text{ uniform on }\mathcal{S}_+^n,
\end{split}
\end{equation}
and $\mathcal{S}_+^n$ is the unit hyper-sphere restricted to the positive orthant
given by
\begin{equation}\label{S_intro}
\mathcal{S}_+^n=\left\{\bm{h}: \|\bm{h}\|^2=h_1^2+\dots+h_n^2=1, \mbox{ and } h_i\geq 0, \ (i=1,\dots,n)\right\}.
\end{equation}
As seen in Remark \ref{rem:MC}, the hierarchical structure of the posterior distribution given by \eqref{mc_step_0_intro} 
enables direct Monte Carlo sample generation essentially based on Gaussian random samplers.
Further, the posterior mean of $\bm{\beta}$ 
has the closed form 
\begin{equation*}
\frac{2^{1/2}\Gamma(\{n+1\}/2)}{\Gamma(n/2)}
 (\bm{X}'_{\bm{y}}\bm{X}_{\bm{y}})^{-1}\bm{X}'_{\bm{y}}
\frac{\displaystyle E_{\bm{h}}
\left[\bm{h}\|\bm{\Psi}[\bm{X}_{\bm{y}}]\bm{h}\|^{-(n+1)}\right]}
{\displaystyle E_{\bm{h}}
\left[\|\bm{\Psi}[\bm{X}_{\bm{y}}]\bm{h}\|^{-n}\right]}
\end{equation*}
where 
$ \mathrm{E}_{\bm{h}}$ refers to the expectation with respect to the uniform 
distribution on $\mathcal{S}_+^n$.
%\begin{equation}
%%\begin{split}
%%& 
%(\bm{X}'_{\bm{y}}\bm{X}_{\bm{y}})^{-1}\bm{X}'_{\bm{y}}
%\bm{w}[\bm{X}_{\bm{y}}] %\\
%%& \quad =(\bm{X}'\bm{X})^{-1}\bm{X}'
%%\left(2\mathrm{diag}(\bm{y})-\bm{I}_n\right)\bm{w}[\bm{X}_{\bm{y}}] 
%%\end{split}
%\end{equation}
%where 
%\begin{equation}\label{intro_z}
% \bm{w}[\bm{X}_{\bm{y}}]=\frac{2^{1/2}\Gamma(\{n+1\}/2)}{\Gamma(n/2)}
%\frac{\displaystyle E_{\bm{h}}
%\left[\bm{h}\|\bm{\Psi}[\bm{X}_{\bm{y}}]\bm{h}\|^{-(n+1)}\right]}
%{\displaystyle E_{\bm{h}}
%\left[\|\bm{\Psi}[\bm{X}_{\bm{y}}]\bm{h}\|^{-n}\right]}.
%\end{equation}
%Clearly $ \left(2\mathrm{diag}(\bm{y})-\bm{I}_n\right)\bm{w}[\bm{X}_{\bm{y}}] $
%is interpreted as a response variable for Bayesian probit regression.
More generally, a closed form of any moment of the posterior distribution,
including the posterior variance, can also be expressed similarly.

This paper is organized as follows. In Section \ref{sec:polar},
we give a polar coordinate representation of the joint probability given by
\eqref{joint_1}. Using this representation, we develop
an alternative representation of the posterior distribution
in Section \ref{sec:post_mean}, which leads to more efficient direct simulation
based analyses.

\section{The polar coordinate representation of the joint probability}
\label{sec:polar}
The probability that $Y_i=1$ is given by
\begin{equation}\label{probit_1}
\begin{split}
\mathrm{Pr}(Y_i=1|\bm{\beta})&=\Phi(\bm{x}'_i\bm{\beta}) 
  =\int_{-\infty}^{\bm{x}'_i\bm{\beta}}
\frac{1}{(2\pi)^{1/2}}\exp(-t^2/2)dt \\
 & =\int_{0}^{\infty}
\frac{1}{(2\pi)^{1/2}}\exp\left(-\frac{(t-\bm{x}'_i\bm{\beta})^2}{2}\right)dt .
\end{split}
\end{equation}
Similarly we have
\begin{equation}\label{probit_-1}
 \mathrm{Pr}(Y_i=0|\bm{\beta})%=1-\Phi(\bm{x}'_i\bm{\beta}) 
=\Phi(-\bm{x}'_i\bm{\beta}) 
 =\int_{0}^{\infty}
\frac{1}{(2\pi)^{1/2}}\exp\left(-\frac{(t+\bm{x}'_i\bm{\beta})^2}{2}\right)dt .
\end{equation}
By \eqref{probit_1} and \eqref{probit_-1},
\begin{equation*}\label{probit_general}
\mathrm{Pr}(Y_i=y_i|\bm{\beta})
 =\int_{0}^{\infty}
\frac{1}{(2\pi)^{1/2}}\exp\left(-\frac{(t-\{2y_i-1\}\bm{x}'_i\bm{\beta})^2}{2}\right)dt .
\end{equation*}
Since $Y_1,\dots,Y_n$ are mutually independent, the joint probability is
\begin{equation}\label{joint_prob}
 \mathrm{Pr}(\bm{Y}=\bm{y}|\bm{\beta}) 
%=\prod_{i=1}^n  \mathrm{Pr}(Y_i=y_i|\bm{\beta}) \\
=\int_{\mathcal{R}^n_+} \frac{1}{(2\pi)^{n/2}}
\exp\left(-\frac{\|\bm{t}-\bm{X}_{\bm{y}}\bm{\beta}\|^2}{2}\right)d\bm{t} 
\end{equation}
where $ \bm{X}_{\bm{y}}=\left\{2\mathrm{diag}(\bm{y})-\bm{I}_n\right\}\bm{X}$,
$\bm{t}=(t_1,\dots,t_n)'$ and 
the range of integration is the positive orthant of $\mathcal{R}^n $
given by
\begin{equation}\label{range_t}
\mathcal{R}^n_+=\{\bm{t}| 0<t_i<\infty \ (1\leq i\leq n)\}.
\end{equation}
Note that by the presentation of \eqref{joint_prob},
$\bm{X}_{\bm{y}}$ is sufficient for the joint probability.

A polar coordinate representation of \eqref{joint_prob}
for $t_1,\dots,t_n$ is given by
\begin{equation}\label{polar}
\begin{split}
& \textstyle{t_1=s^{1/2}\cos\varphi_1, \
t_i=s^{1/2}\prod_{j=1}^{i}\sin\varphi_{j}\cos\varphi_i, \ (i=2,\dots, n-1),} \\
&\textstyle{t_n=s^{1/2}\prod_{j=1}^{n-1}\sin\varphi_{j} ,\qquad
 \bm{t}=s^{1/2}\bm{h}(\bm{\varphi})}
\end{split}
\end{equation}
where $\bm{\varphi}=(\varphi_1,\dots,\varphi_{n-1})'$.
The Jacobian is 
\begin{equation}\label{Jacob_polar}
\textstyle{ \mathrm{Jacobian}[\bm{t}\to(s,\bm{\varphi}')'] =2^{-1}s^{n/2-1}
\prod_{j=1}^{n-2}\left\{\sin\varphi_j\right\}^{n-1-j}}.
\end{equation}
From \eqref{range_t}, 
the range of $\bm{\varphi}$, $\mathrm{R}(\bm{\varphi})$,
is given by
\begin{equation*}\label{range_varphi}
 0<\varphi_i<\pi/2 \ (i=1,\dots,n-1).
\end{equation*}
Therefore we have
\begin{equation}\label{joint_prob_polar}
\mathrm{Pr}(\bm{Y}=\bm{y}|\bm{\beta}) 
=\frac{1}{2(2\pi)^{n/2}}
\int_{\mathrm{R}(\bm{\varphi})} m(\bm{y}|\bm{\beta},\bm{h}(\bm{\varphi}))
\prod_{j=1}^{n-2}\left\{\sin\varphi_j\right\}^{n-1-j} d\bm{\varphi}
\end{equation}
where
\begin{equation*}\label{m_0}
m(\bm{y}|\bm{\beta},\bm{h}(\bm{\varphi}))
=\int_{0}^\infty s^{n/2-1}
\exp\left(-\frac{\|s^{1/2}\bm{h}(\bm{\varphi})-\bm{X}_{\bm{y}}\bm{\beta}\|^2}{2}\right) ds.\end{equation*}
Note that
\begin{equation*}\label{beta_sin}
\int_0^{\pi/2}\left\{\sin\varphi_j\right\}^{n-1-j} d\varphi
=\frac{B(1/2,\{n-j\}/2)}{2}
=\frac{\pi^{1/2}\Gamma(\{n-j\}/2)}{2\Gamma(\{n-j+1\}/2)},
\end{equation*}
and that
\begin{equation*}\label{beta_sin_1}
\int_{\mathrm{R}(\bm{\varphi})} \prod_{j=1}^{n-2}
\left\{\sin\varphi_j\right\}^{n-1-j} d\bm{\varphi} 
=\frac{\pi^{n/2}}{2^{n}\Gamma(n/2)}=c_1(n).
\end{equation*}
Therefore the joint probability is given by
\begin{equation}\label{joint_prob_polar_1}
\mathrm{Pr}(\bm{Y}=\bm{y}|\bm{\beta}) 
=\frac{c_1(n)}{2(2\pi)^{n/2}}
\mathrm{E}_{\bm{h}}\left[m(\bm{y}|\bm{\beta},\bm{h}) \right].
\end{equation}
In \eqref{joint_prob_polar_1}, 
$ \mathrm{E}_{\bm{h}}$ refers to the expectation with respect to the distribution of
$\bm{h}=(h_1,\dots,h_n)'$, which is uniformly distributed on $\mathcal{S}_+^n$, given in
\eqref{S_intro}, the unit hyper-sphere restricted to the positive orthants.

\section{Posterior inference with respect to the flat prior}
\label{sec:post_mean}
In this section, we consider posterior inference with respect to the flat prior.
The posterior distribution is given by
\begin{equation}\label{posterior_distribution_1}
\pi(\bm{\beta}|\bm{y})= \frac
{E_{\bm{h}}[m(\bm{y}|\bm{\beta},\bm{h})] }
{E_{\bm{h}}[\int_{\mathcal{R}^p}m(\bm{y}|\bm{\beta},\bm{h})
d\bm{\beta}] }.
\end{equation}
First we give a hierarchical structure for the posterior distribution which 
enables simple and efficient Monte Carlo sample generation.
Recall that $ \bm{\Psi}[\bm{X}_{\bm{y}}]$ given in \eqref{Psi_intro}
is the projection matrix to the orthogonal
complement of the column space of $\bm{X}_{\bm{y}}$.

For posterior inference, propriety of posteriors 
has been studied by many. % including Speckman et al.~(2009). 
\cite{Speckman-Jaeyong-Sun-2009} showed that the posterior distribution with respect to
the flat prior is proper 
if and only if 
%for any choice of $\bm{\beta}$, there is an
%$\bm{x}_i$ such that 
%\begin{equation}
% (2y_i-1)\bm{x}'_i\bm{\beta}<0.
%\end{equation}
%Or equivalently, 
there does not exist $\bm{\beta}$ such that
\begin{equation}
 (2\mathrm{diag}(\bm{y})-\bm{I}_n)\bm{X}\bm{\beta}=\bm{X}_{\bm{y}}\bm{\beta}\in \mathcal{R}^n_+.
\end{equation}
As seen in Lemma \ref{lem:equiv} below, 
this is equivalent to 
the non-existence of $\bm{u}\in \mathcal{S}_+^n$ 
such that $ \bm{\Psi}[\bm{X}_{\bm{y}}]\bm{u}=\bm{0}$.
So $\pi(s|\bm{u},\bm{y})$ and $\pi(\bm{u}|\bm{y})$ below are well-defined.
\begin{thm}\label{thm:main}
Assume 
there does not exist $\bm{u}\in \mathcal{S}_+^n$
such that $ \bm{\Psi}[\bm{X}_{\bm{y}}]\bm{u}=\bm{0}$. 
Then $\pi(\bm{\beta}|\bm{y})$ is given by 
\begin{equation}
\pi(\bm{\beta}|\bm{y})=
\iint \pi(\bm{\beta}|s,\bm{u},\bm{y})\pi(s|\bm{u},\bm{y})\pi(\bm{u}|\bm{y})
dsd\bm{u}
\end{equation}
where the elements of this hierarchical structure are given by
\begin{equation}\label{mc_step_0}
\begin{split}
 \pi(\bm{\beta}|s,\bm{u},\bm{y})&=N_p((\bm{X}'\bm{X})^{-1}\bm{X}'_{\bm{y}}s^{1/2}\bm{u},
(\bm{X}'\bm{X})^{-1}),\\
\pi(s|\bm{u},\bm{y}) &=\|\bm{\Psi}[\bm{X}_{\bm{y}}]\bm{u}\|^{-2}\chi_n^2, \\
\pi(\bm{u}|\bm{y})& \propto \|\bm{\Psi}[\bm{X}_{\bm{y}}]\bm{u}\|^{-n} \text{ on }\mathcal{S}_+^n.
\end{split}
\end{equation}
\end{thm}
\begin{proof}%[Proof of Theorem \ref{thm:main}]
Note $ \bm{X}'_{\bm{y}}\bm{X}_{\bm{y}}=\bm{X}'\bm{X}$. 
Since
\begin{equation}\label{complete_squares}
 \|s^{1/2}\bm{h}-\bm{X}_{\bm{y}}\bm{\beta}\|^2 = 
s\|\bm{\Psi}[\bm{X}_{\bm{y}}]\bm{h}\|^2  +
(\bm{\beta}-\hat{\bm{\beta}})'
\bm{X}'\bm{X}
(\bm{\beta}-\hat{\bm{\beta}}),
\end{equation}
where
\begin{equation*}
 \hat{\bm{\beta}}=s^{1/2}(\bm{X}'\bm{X})^{-1}\bm{X}'_{\bm{y}}\bm{h},
\end{equation*}
we have
\begin{equation}
 \begin{split}
& E_{\bm{h}}[m(\bm{y}|\bm{\beta},\bm{h})]
=E_{\bm{h}}\left[\int_0^\infty s^{n/2-1}
\exp\left(-s\frac{\|\bm{\Psi}[\bm{X}_{\bm{y}}]\bm{h}\|^2}{2}\right)\right.
   \\
&\qquad \left.\times \exp\left(-\frac{
(\bm{\beta}-\hat{\bm{\beta}})'\bm{X}'\bm{X}
(\bm{\beta}-\hat{\bm{\beta}})}{2}\right)ds \right] \\
&= \frac{(2\pi)^{p/2}}{|\bm{X}'\bm{X}|^{1/2}}2^{n/2}\Gamma(n/2)E_{\bm{h}}
\left[\frac{1}{\|\bm{\Psi}[\bm{X}_{\bm{y}}]\bm{h}\|^{n}} \right.\\
&\quad \times \int_0^\infty \frac{\|\bm{\Psi}[\bm{X}_{\bm{y}}]\bm{h}\|^{n}}{2^{n/2}\Gamma(n/2)}s^{n/2-1}\exp\left(-s\frac{\|\bm{\Psi}[\bm{X}_{\bm{y}}]\bm{h}\|^2}{2}\right) \\
& \qquad \left.
\times \frac{|\bm{X}'\bm{X}|^{1/2}}{(2\pi)^{p/2}}\exp\left(-\frac{
(\bm{\beta}-\hat{\bm{\beta}})'\bm{X}'\bm{X}
(\bm{\beta}-\hat{\bm{\beta}})}{2}\right)ds \right]
\end{split}
\end{equation}
provided all integrals exist. 
Let 
\begin{equation}
 \pi(\bm{u}|\bm{y})=\frac{\|\bm{\Psi}[\bm{X}_{\bm{y}}]\bm{u}\|^{-n}}
{\int_{\bm{u}\in\mathcal{S}^n_+}\|\bm{\Psi}[\bm{X}_{\bm{y}}]\bm{u}\|^{-n}d\bm{u} }.
\end{equation}
Then 
\begin{equation}
\begin{split}
 E_{\bm{h}}[m(\bm{y}|\bm{\beta},\bm{h})]&=
\frac{(2\pi)^{p/2}2^{n/2}\Gamma(n/2)}{|\bm{X}'\bm{X}|^{1/2}}
\int_{\bm{u}\in\mathcal{S}^n_+}\|\bm{\Psi}[\bm{X}_{\bm{y}}]\bm{u}\|^{-n}d\bm{u} \\
& \quad \times\int_0^\infty\int_{\mathcal{S}_+^n}  
 \pi(\bm{\beta}|s,\bm{u},\bm{y})\pi(s|\bm{u},\bm{y})\pi(\bm{u}|\bm{y})ds d\bm{u}.
\end{split}
\end{equation}
Hence the theorem follows.
\end{proof}
\begin{remark}\label{rem:MC}
Let $Z_1,\dots,Z_n$ be independently distributed $N(0,1)$.
Then $ t=\sum_{i=1}^n Z^2_i\sim \chi_n^2$ and
$ \bm{h}=(|Z_1|/t^{1/2},\dots,|Z_n|/t^{1/2})'$
is uniformly distributed on $\mathcal{S}_+^n$ which is independent of $t$.
Using this property, we can propose the following algorithm.
The so-called SIR (Sampling/Importance Resampling) method, 
described in part 2 and 3 below,
enables to obtain samples from $\pi(\bm{u}|\bm{y})$ based on samples from 
the uniform distribution on $\mathcal{S}_+^n$.
%we can propose the following algorithm.
%Let $s=t/\|\bm{\Psi}[\bm{X}_{\bm{y}}]\bm{h}\|^2$.
%Hence $ (s, \bm{h})$ has the joint distribution $\pi(s|\bm{h},\bm{y})\pi(\bm{h}|\bm{y})$
%given by \eqref{mc_step_0}.
%Note $s^{1/2}\bm{h}$ is written as
%\begin{equation*}
% s^{1/2}\bm{h}=
%\frac{t^{1/2}}{\|\bm{\Psi}[\bm{X}_{\bm{y}}]\bm{h}\|}
%\bm{h}
%=
%\frac{\|t^{1/2}\bm{h}\|}{\|\bm{\Psi}[\bm{X}_{\bm{y}}]\{t^{1/2}\bm{h}\}\|}
%\{t^{1/2}\bm{h}\}.
%\end{equation*}
\end{remark}
\begin{algorithm}[sampling from the posterior distribution $\pi(\bm{\beta}|\bm{y})$]\mbox{}\label{alg:sampling}
\begin{enumerate}
\item For $i=1,\dots,N$
\begin{enumerate}
\item Generate $n$ standard Normal random samples $z_1^{(i)},\dots,z_n^{(i)}$.
\item Compute $t^{(i)}=\sum_{j=1}^n \{z_j^{(i)}\}^2$ and
$\bm{h}^{(i)}=\{t^{(i)}\}^{-1/2}(|z_1^{(i)}|,\dots,|z_n^{(i)}|)'$.
%Take the absolute value of them $u_i=|z_i| $ for $i=1,\dots,n$.
\item Compute $v^{(i)}=1/\|\bm{\Psi}[\bm{X}_{\bm{y}}]\bm{h}^{(i)}\| $.
\end{enumerate}
\item Re-sample $i_1,\dots,i_{M}$ with replacement from the multinomial distribution
with probabilities 
\begin{equation*}
\mathrm{Pr}(J=i)= \frac{\{v^{(i)}\}^{n}}{\sum_{i=1}^{N} \{v^{(i)}\}^{n}}, \ i=1,\dots,N.
\end{equation*}
\item Let $ \bm{u}^{(k)}=\bm{h}^{(i_k)}$ and $w^{(k)}=v^{(i_k)}$ for $k=1,\dots,M$.
\item For $k=1,\dots,M$
\begin{enumerate}
\item Generate $p$ standard Normal random samples $z_1^{(k)},\dots,z_p^{(k)}$.
\item Compute
\begin{equation*}
 \bm{\beta}^{(k)}=w^{(k)}\sqrt{t^{(k)}}(\bm{X}'\bm{X})^{-1}\bm{X}'_{\bm{y}}\bm{u}^{(k)}+
(\bm{X}'\bm{X})^{-1/2}(z_{1}^{(k)},\dots,z_{p}^{(k)})'.
\end{equation*}
\end{enumerate}
%\item $\beta_i=u_i/v+\sum_{k=1}^p sigma_$
%\item Repeat $1$-$5$.
\end{enumerate}
\end{algorithm}

Therefore, in order to obtain samples from \eqref{mc_step_0},
it suffices to generate the Gaussian random samples.
The simplicity and directness of this method gives an advantage over the methods in much of the literature.
In addition to the simple structure described above, we note that, generally speaking,
Monte Carlo sampling is  more efficient than
MCMC sampling which has been utilized in this area.

\begin{remark}
When interest lies primary in 
the posterior moments, a closed form of any posterior moments with respect to
the flat prior is available.
For example, the posterior mean is given by
\begin{equation}\label{eq:posterior_mean_part_0}
 \mathrm{E}[\bm{\beta}|\bm{y}]
=\frac{\int \bm{\beta}\pi(\bm{\beta}|\bm{y})d\bm{\beta}}
{\int \pi(\bm{\beta}|\bm{y})d\bm{\beta}}
=(\bm{X}'\bm{X})^{-1}\bm{X}'_{\bm{y}}\mathrm{E}_{s,\bm{u}}[s^{1/2}\bm{u}|\bm{y}]
\end{equation}
where $\mathrm{E}[s^{1/2}\bm{u}|\bm{y}]$ may be written as
\begin{equation} \label{eq:posterior_mean_part_1}
\mathrm{E}_{\bm{u}}\left[\mathrm{E}_{s|\bm{u}}
\left[s^{1/2}|\bm{u},\bm{y}\right]\bm{u}\right] 
=\frac{2^{1/2}\Gamma(\{n+1\}/2)}{\Gamma(n/2)}
\frac{\displaystyle E_{\bm{h}}
\left[\bm{h}\|\bm{\Psi}[\bm{X}_{\bm{y}}]\bm{h}\|^{-(n+1)}\right]}
{\displaystyle E_{\bm{h}}
\left[\|\bm{\Psi}[\bm{X}_{\bm{y}}]\bm{h}\|^{-n}\right]}.
\end{equation}
The posterior variance is
\begin{equation}\label{eq:pos_var}
\begin{split}
& \mathrm{Var}[\bm{\beta}|\bm{y}]=\mathrm{E}\left[(\bm{\beta}-\mathrm{E}[\bm{\beta}|\bm{y}])
(\bm{\beta}-\mathrm{E}[\bm{\beta}|\bm{y}])'|\bm{y}\right] \\
&  
= \mathrm{E}\left[\bm{\beta}\bm{\beta}'|\bm{y}\right]-\mathrm{E}[\bm{\beta}|\bm{y}]\mathrm{E}[\bm{\beta}|\bm{y}]'\\
&=\mathrm{E}_{s,\bm{u}}\left[\mathrm{E}\left[\bm{\beta}\bm{\beta}'|s,\bm{u},\bm{y}\right]
\right]-\mathrm{E}[\bm{\beta}|\bm{y}]\mathrm{E}[\bm{\beta}|\bm{y}]' \\
&=(\bm{X}'\bm{X})^{-1}+
%\mathrm{E}_{s,\bm{u}}\left[\mathrm{Var}\left[\bm{\beta}|s,\bm{u},\bm{y}\right]\right] +
\mathrm{E}_{s,\bm{u}}\left[\mathrm{E}\left[\bm{\beta}|s,\bm{u},\bm{y}\right]
\mathrm{E}\left[\bm{\beta}|s,\bm{u},\bm{y}\right]'\right] 
 -\mathrm{E}[\bm{\beta}|\bm{y}]\mathrm{E}[\bm{\beta}|\bm{y}]' 
%&=(\bm{X}'\bm{X})^{-1}
%+\mathrm{E}_{s,\bm{h}}\left[s(\bm{X}'\bm{X})^{-1}\bm{X}'_{\bm{y}}\bm{h}\bm{h}'\bm{X}_{\bm{y}}(\bm{X}'\bm{X})^{-1}
%|s,\bm{h},\bm{y}\right] \\& \qquad \qquad -\mathrm{E}[\bm{\beta}|\bm{y}]\mathrm{E}[\bm{\beta}|\bm{y}]' \\
%&=(\bm{X}'\bm{X})^{-1}+n
%(\bm{X}'\bm{X})^{-1}\bm{X}'_{\bm{y}}
%\frac{\displaystyle \mathrm{E}_{\bm{h}}
%\left[\frac{\bm{h}\bm{h}'}{\|\bm{\Psi}[\bm{X}_{\bm{y}}]\bm{h}\|^{n+2}}\right]}
%{\displaystyle \mathrm{E}_{\bm{h}}
%\left[\|\bm{\Psi}[\bm{X}_{\bm{y}}]\bm{h}\|^{-n}\right]}
%\bm{X}_{\bm{y}}(\bm{X}'\bm{X})^{-1} \\
%& \qquad \qquad -\mathrm{E}[\bm{\beta}|\bm{y}]\mathrm{E}[\bm{\beta}|\bm{y}]'
\end{split}
\end{equation}
where the second term of the right-hand side of \eqref{eq:pos_var} is re-expressed as
\begin{equation}\label{eq:2_posterior_var}
\begin{split}
& \mathrm{E}_{s,\bm{u}}\left[\mathrm{E}\left[\bm{\beta}|s,\bm{u},\bm{y}\right]
\mathrm{E}\left[\bm{\beta}|s,\bm{u},\bm{y}\right]'\right] \\
&=\mathrm{E}_{s,\bm{u}}\left[s(\bm{X}'\bm{X})^{-1}\bm{X}'_{\bm{y}}\bm{u}\bm{u}'\bm{X}_{\bm{y}}(\bm{X}'\bm{X})^{-1}|s,\bm{u},\bm{y}\right] \\
&=n
(\bm{X}'\bm{X})^{-1}\bm{X}'_{\bm{y}}
\frac{\displaystyle \mathrm{E}_{\bm{h}}
\left[\bm{h}\bm{h}'\|\bm{\Psi}[\bm{X}_{\bm{y}}]\bm{h}\|^{-(n+2)}\right]}
{\displaystyle \mathrm{E}_{\bm{h}}
\left[\|\bm{\Psi}[\bm{X}_{\bm{y}}]\bm{h}\|^{-n}\right]}
\bm{X}_{\bm{y}}(\bm{X}'\bm{X})^{-1}.
\end{split} 
\end{equation}
Compared to the sample mean and sample variance of
simulated samples by Algorithm \ref{alg:sampling}, %from the posterior distribution, 
these expressions with closed forms given in 
\eqref{eq:posterior_mean_part_0},\eqref{eq:posterior_mean_part_1}, 
\eqref{eq:pos_var} and \eqref{eq:2_posterior_var} 
should be more useful and efficient, where samples $\bm{h}^{(i)}$ for $i=1,\dots,N$
in Algorithm \ref{alg:sampling} are sufficient.

In order to give an explicit expression of higher order moments,
the first step is to take the expectation of the functions of $\bm{\beta}$
given $s$ and $\bm{u}$,
as in the third equation of the right-hand side of \eqref{eq:pos_var}.
These are the moments of multivariate Normal distribution.
The second step is take the expectation of the function of $s$ given $\bm{u}$.
These are the moments of $\chi^2$ distribution.
As a result, the expression of the moments while complicated, are still exact.

Further we note that 
the posterior moments in the above have a kind of equivariant property.
Suppose that the scale of each covariate changes as $\bm{XD}$ with 
a $p\times p$ positive diagonal matrix $\bm{D}$.
Since
\begin{equation}
 \begin{split}
&  \bm{\Psi}[\{\bm{XD}\}_{\bm{y}}]=\bm{\Psi}[\bm{X}_{\bm{y}}], \\
& (\{\bm{XD}\}_{\bm{y}}'\{\bm{XD}\}_{\bm{y}})^{-1}\{\bm{XD}\}_{\bm{y}}'
=\bm{D}^{-1}(\bm{X}'\bm{X})^{-1}\bm{X}'_{\bm{y}},
 \end{split}
\end{equation}
the posterior means satisfy the desirable property
\begin{equation}
 \mathrm{E}[\bm{\beta}|\{\bm{XD}\}_{\bm{y}}]=
\bm{D}^{-1}\mathrm{E}[\bm{\beta}|\bm{X}_{\bm{y}}].
\end{equation}
In the same way, the posterior variances satisfy
\begin{equation}
 \mathrm{Var}[\bm{\beta}|\{\bm{XD}\}_{\bm{y}}]
=\bm{D}^{-1}\mathrm{Var}[\bm{\beta}|\bm{X}_{\bm{y}}]\bm{D}^{-1}.
\end{equation}
\end{remark}
\begin{remark}
We briefly remark that a similar direct MC sampler is possible 
for Normal priors for $\bm{\beta}$.
Let $\bm{\beta}\sim N_p(\bm{0},\bm{Q})$.
The main difference comes from completing squares corresponding to
\eqref{complete_squares} as
\begin{equation}
\begin{split}
& \bm{\beta}'\bm{Q}^{-1}\bm{\beta}+ \|s^{1/2}\bm{h}-\bm{X}_{\bm{y}}\bm{\beta}\|^2 \\
&\qquad =
s\bm{h}'\bm{\Psi}[\bm{X}_{\bm{y}},\bm{Q}]\bm{h}
+(\bm{\beta}-\hat{\bm{\beta}})'
(\bm{X}'\bm{X}+\bm{Q}^{-1})
(\bm{\beta}-\hat{\bm{\beta}}),
\end{split}
\end{equation}
where
\begin{equation*}
 \bm{\Psi}[\bm{X}_{\bm{y}},\bm{Q}]=\bm{I}_n-\bm{X}_{\bm{y}}(\bm{X}'\bm{X}+\bm{Q}^{-1})^{-1}\bm{X}'_{\bm{y}}
\end{equation*}
\begin{equation*}
 \hat{\bm{\beta}}=s^{1/2}(\bm{X}'\bm{X}+\bm{Q}^{-1})^{-1}\bm{X}'_{\bm{y}}\bm{h}.
\end{equation*}
Hence there is a corresponding hierarchical structure to \eqref{mc_step_0} 
of Theorem \ref{thm:main},
which is given by
\begin{equation*}
\begin{split}
 \pi(\bm{\beta}|s,\bm{u},\bm{y})&=N_p((\bm{X}'\bm{X}+\bm{Q}^{-1})^{-1}\bm{X}'_{\bm{y}}s^{1/2}\bm{u},
(\bm{X}'\bm{X}+\bm{Q}^{-1})^{-1}),\\
\pi(s|\bm{u},\bm{y}) &=(\bm{u}'\bm{\Psi}[\bm{X}_{\bm{y}},\bm{Q}]\bm{u})^{-1}\chi_n^2, \\
\pi(\bm{u}|\bm{y})& \propto (\bm{u}'\bm{\Psi}[\bm{X}_{\bm{y}},\bm{Q}]\bm{u})^{-n/2}\text{  on }\mathcal{S}_+^n.
\end{split}
\end{equation*}
Since the prior is proper, the posterior is always proper even if 
there  exists $\bm{\alpha}\in\mathcal{R}^p$ such that $\bm{X_{\bm{y}}}\bm{\alpha}\in \mathcal{R}_+^n$. 
\end{remark}

\bigskip

The lemma below is related to the regularity condition of Theorem \ref{thm:main}.
\begin{lemma}\label{lem:equiv}
 The necessary and sufficient condition for the existence of
$\bm{\alpha}\in\mathcal{R}^p$ such that $\bm{X_{\bm{y}}}\bm{\alpha}\in \mathcal{R}_+^n$ is
that the existence of $\bm{u}\in\mathcal{S}^n_+$ 
such that $\bm{\Psi}[\bm{X}_{\bm{y}}]\bm{u}=\bm{0}$.
\end{lemma}
\begin{proof}
 Suppose there exists 
 $\bm{\alpha}\in\mathcal{R}^p$ such that $\bm{X\alpha}\in \mathcal{R}_+^n$.
Let $\bm{u}=\bm{X\alpha}$. 
Then $ \bm{u}$ satisfies $\bm{\Psi}[\bm{X}_{\bm{y}}]\bm{u}=\bm{0}$.

Suppose there exists $\bm{u}\in\mathcal{S}^n_+$ 
such that $\bm{\Psi}[\bm{X}_{\bm{y}}]\bm{u}=\bm{0}$. 
Let $\bm{\alpha}=(\bm{X}'\bm{X})^{-1}\bm{X}'_{\bm{y}}\bm{u}$.
Then $\bm{X\alpha}=\bm{X}_{\bm{y}}(\bm{X}'\bm{X})^{-1}\bm{X}'_{\bm{y}}\bm{u}=\bm{u}$ is in
$\mathcal{R}_+^n$.
\end{proof}


\begin{thebibliography}{3}
% BibTex style file: imsart-nameyear.bst, 2010-01-14
% Default style options (sort=1,type=nameyear).
% Used options (sort=1,type=nameyear).

\bibitem[\protect\citeauthoryear{Albert and Chib}{1993}]{Albert-Chib-1993}
\begin{barticle}[author]
\bauthor{\bsnm{Albert},~\bfnm{James~H.}\binits{J.~H.}} \AND
  \bauthor{\bsnm{Chib},~\bfnm{Siddhartha}\binits{S.}}
(\byear{1993}).
\btitle{Bayesian analysis of binary and polychotomous response data}.
\bjournal{J. Amer. Statist. Assoc.}
\bvolume{88}
\bpages{669--679}.
\bmrnumber{1224394}
\end{barticle}
\endbibitem

\bibitem[\protect\citeauthoryear{Marin and Robert}{2007}]{Marin-Robert-2007}
\begin{bbook}[author]
\bauthor{\bsnm{Marin},~\bfnm{Jean-Michel}\binits{J.-M.}} \AND
  \bauthor{\bsnm{Robert},~\bfnm{Christian~P.}\binits{C.~P.}}
(\byear{2007}).
\btitle{Bayesian core: a practical approach to computational {B}ayesian
  statistics}.
\bseries{Springer Texts in Statistics}.
\bpublisher{Springer}, \baddress{New York}.
\bmrnumber{2289769}
\end{bbook}
\endbibitem

\bibitem[\protect\citeauthoryear{Speckman, Lee and
  Sun}{2009}]{Speckman-Jaeyong-Sun-2009}
\begin{barticle}[author]
\bauthor{\bsnm{Speckman},~\bfnm{Paul~L.}\binits{P.~L.}},
  \bauthor{\bsnm{Lee},~\bfnm{Jaeyong}\binits{J.}} \AND
  \bauthor{\bsnm{Sun},~\bfnm{Dongchu}\binits{D.}}
(\byear{2009}).
\btitle{Existence of the {MLE} and propriety of posteriors for a general
  multinomial choice model}.
\bjournal{Statist. Sinica}
\bvolume{19}
\bpages{731--748}.
\bmrnumber{2514185}
\end{barticle}
\endbibitem

\end{thebibliography}
\end{document}